\documentclass[11pt,reqno]{amsart}
\usepackage[T1]{fontenc}
\usepackage[utf8]{inputenc}
\usepackage{amsmath,amssymb,amsthm}
\usepackage{booktabs}
\usepackage[margin=1in]{geometry}
\usepackage[colorlinks=true,linkcolor=blue,citecolor=blue,urlcolor=blue]{hyperref}

\theoremstyle{plain}
\newtheorem{theorem}{Theorem}[section]
\newtheorem{proposition}[theorem]{Proposition}
\newtheorem{lemma}[theorem]{Lemma}
\newtheorem{corollary}[theorem]{Corollary}
\theoremstyle{definition}
\newtheorem{problem}[theorem]{Problem}
\newtheorem{remark}[theorem]{Remark}
\newtheorem{example}[theorem]{Example}

\newcommand{\T}{\mathbb{T}}
\newcommand{\Zed}{\mathbb{Z}}
\newcommand{\Rat}{\mathbb{Q}}
\newcommand{\Real}{\mathbb{R}}
\newcommand{\Disk}{\mathbb{D}}
\DeclareMathOperator{\Rp}{Re}
\DeclareMathOperator{\Ip}{Im}
\DeclareMathOperator{\cone}{cone}

\begin{document}

\title[Conservation of Normalized Energy Density]
{Conservation of Normalized Energy Density in multi-layered media:\\
a proof for the general multi-layer 2D SH problem}

\author{Hiroyuki Goto}
\address{Disaster Prevention Research Institute, Kyoto University,
Gokasho, Uji, Kyoto 611-0011, Japan}
\email{goto.hiroyuki.3z@kyoto-u.ac.jp}


\date{\today}

\begin{abstract}
The Normalized Energy Density (NED) was introduced in Goto et al.(2011)
as a quantity that, unlike the conventional energy, is conserved across material
interfaces in a layered elastic medium. Its conservation was proved analytically
only for a two-layer structure; for three or more layers it was supported by
Monte Carlo simulation. I supply the missing proof for a general multi-layer ($n$-layer)
structure in the 2D SH problem. The key step is to stop treating the angular
frequency $\omega$ as a single variable and instead lift the layer phases
$\theta_k=\omega\tau_k$ to independent coordinates on the $n$-torus.
The obstruction identified in the original paper---that the ratio of travel-time
combinations is generically irrational, so that the periodicity argument used for
two layers fails---then disappears, and the proof reduces to Weyl's equidistribution
theorem combined with $n$ repetitions of the same elementary integral used in the
two-layer case. I prove conservation in two independent settings that together
cover all cases of practical interest: travel times that are linearly independent
over $\Rat$, and commensurable travel times, including the
equal-travel-time (Goupillaud) discretization. I also show that an arbitrary layered structure can be
subdivided by reflectionless interfaces into one whose travel times take only
$d$ distinct, $\Rat$-independent values, where $d$ is the dimension of the
$\Rat$-span of the original travel times; this places the two proved cases at the
endpoints $d=1$ and $d=n$ and isolates the remaining open case
$1<d<n$ in a canonical form.
\end{abstract}

\maketitle


\section{Introduction}

Conserved quantities are the basic tools for analysing wave propagation in a
continuous medium, and the balance of energy in particular underlies the
quantification of seismic energy. For a layered structure, however, the
conventional energy is not conserved across a material interface. A part of
the incident energy is transmitted and part is reflected, so the energy observed
in one layer is only a portion of the incident energy. This is a familiar point,
although quantities of the form $\int\rho c\,\dot u^2\,dt$ are sometimes applied
to layered structures without noting that they are not conserved
\cite{KokushoMotoyama}.

In Goto et al.(2011)\cite{GSH2011}, a different quantity, the Normalized Energy Density
(NED), was introduced for the 2D SH problem. It is defined as the frequency
average of the squared modulus of half the transfer function, multiplied by the
impedance, and it is conserved across interfaces. Its practical appeal is
that it can be evaluated from the surface transfer function and the impedance of
the uppermost layer alone, which means no detailed knowledge of the layering is required.
This has been exploited to estimate near-surface damping directly
\cite{Goto2013}, since NED decreases in the presence of attenuation, and it bears
on the choice of site parameters for amplification models.

The conservation of NED was proved in \cite{GSH2011} only for a two-layer
structure. For three, four and ten layers it was examined by Monte Carlo
simulation and found to hold, but an analytical proof was left open. The obstacle
was identified explicitly in \cite[\S3]{GSH2011}: for three layers the relevant
quantity $\alpha_2C_1$ is a sum of functions periodic in
$\omega(\tau_1+\tau_2)$ and $\omega(\tau_1-\tau_2)$, and since
$(\tau_1+\tau_2)/(\tau_1-\tau_2)$ is generically irrational, $\alpha_2C_1$ is not
a periodic function of $\omega$. The single-period integration that settles the
two-layer case is therefore unavailable.

The purpose of this report is to show that this obstacle is an artefact of treating
$\omega$ as a single variable. If instead the layer phases
$\theta_k=\omega\tau_k$ are lifted to independent coordinates on the torus
$\T^n$, the irrationality that blocked the original argument becomes precisely the
hypothesis of Weyl's equidistribution theorem, and the frequency average becomes
an integral over $\T^n$. That integral factorizes, in that integrating over one phase at a
time peels off one layer and produces one factor of the reciprocal impedance
ratio, using the same elementary integral as in the two-layer case. The
telescoping product is the conservation law.

The paper is organised as follows. Section~\ref{sec:setup} fixes notation,
following \cite{GSH2011}. Section~\ref{sec:lemmas} establishes two algebraic
lemmas, of which the second generalizes \cite[Eq.~(30)]{GSH2011} from the
lowermost interface to every interface and yields a uniform positive lower bound
for the denominator. Section~\ref{sec:torus} performs the lifting.
Section~\ref{sec:main} proves the main theorem for $\Rat$-independent travel
times. Section~\ref{sec:goupillaud} treats the equal-travel-time case separately
via orthogonal polynomials on the unit circle; this case is not covered by
Weyl's theorem, and turns out to be a consequence of the classical
Bernstein--Szeg\H{o} normalization. Section~\ref{sec:reduction} shows how
subdivision by reflectionless interfaces extends the second proof to all
commensurable travel times and reduces the general case to a canonical form,
leaving a precisely stated open problem. Sections~\ref{sec:oblique}--\ref{sec:damping}
record the extension to oblique incidence and the situation with attenuation,
Section~\ref{sec:numerics} reports numerical verification, and
Section~\ref{sec:priorwork} discusses the relation to earlier work.

\section{Setting and notation}\label{sec:setup}

I follow \cite[\S3]{GSH2011}. Consider $n$ horizontal layers
$\#1,\dots,\#n$ over a half-space basement $\#0$. Layer $\#k$ has S-wave velocity
$\beta_k$, density $\rho_k$ and thickness $h_k$; the basement has $\beta_0$ and
$\rho_0$. A 2D SH plane wave is incident vertically from the basement. Every
layer is elastic, independently of the wave amplitude, and there is no
attenuation.

Write
\begin{equation}
Z_k=\rho_k\beta_k,\qquad \tau_k=\frac{h_k}{\beta_k}\quad(k=1,\dots,n),
\qquad Z_0=\rho_0\beta_0 .
\end{equation}
To unify notation set $Z_{n+1}:=Z_0$ and let
\begin{equation}
R_k:=\frac{Z_k}{Z_{k+1}}\qquad(k=1,\dots,n),
\end{equation}
so that $R_k=R_{k,k+1}$ for $k\le n-1$ and $R_n=R_{n,0}$ in the notation of
\cite{GSH2011}. Note the telescoping identity
\begin{equation}
R_{k,0}:=\frac{Z_k}{Z_0}=\prod_{j=k}^{n}R_j .
\label{eq:telescope}
\end{equation}
Put
\begin{equation}
p_k=\tfrac12(1+R_k),\qquad q_k=\tfrac12(1-R_k),\qquad \alpha_k=e^{i\omega\tau_k}.
\end{equation}
Since $R_k>0$,
\begin{equation}
p_k>0,\qquad p_k+q_k=1,\qquad p_k-q_k=R_k,\qquad p_k^2-q_k^2=R_k>0,
\label{eq:pq}
\end{equation}
and in particular $p_k>|q_k|$ always.

Let $A_k,B_k$ be the up-going and down-going amplitudes in layer $\#k$, and
$A_0,B_0$ those in the basement. The interface conditions give the propagator
matrix of \cite[Eq.~(21)]{GSH2011}, and the traction-free condition $A_1=B_1$ at
the free surface yields
\begin{equation}
A_k=C_{k-1}A_1,\quad B_k=D_{k-1}A_1\ \ (1\le k\le n),\qquad
A_0=C_nA_1,\quad B_0=D_nA_1,
\end{equation}
where $C_k,D_k$ are defined by
\begin{equation}
\begin{cases}
C_k=p_k\alpha_kC_{k-1}+q_k\alpha_k^{*}D_{k-1},\\[2pt]
D_k=q_k\alpha_kC_{k-1}+p_k\alpha_k^{*}D_{k-1},
\end{cases}
\qquad C_0=D_0=1
\label{eq:recursion}
\end{equation}
(\cite[Eqs.~(27)--(28)]{GSH2011}); the asterisk denotes complex conjugation. The
quantity of interest is
\begin{equation}
P_k(\omega)=\left|\frac{A_k}{A_0}\right|^2=\frac{|C_{k-1}|^2}{|C_n|^2},
\qquad
\langle P_k\rangle=\lim_{\Omega\to\infty}\frac1\Omega
\int_{-\Omega/2}^{\Omega/2}P_k(\omega)\,d\omega ,
\end{equation}
and the NED of layer $\#k$ is $Z_k\langle P_k\rangle$. The assertion to be proved
is
\begin{equation}
Z_k\langle P_k\rangle=Z_0\quad(k=1,\dots,n),
\qquad\text{equivalently}\qquad
\langle P_k\rangle=\frac{1}{R_{k,0}} .
\label{eq:goal}
\end{equation}
That the value is independent of $k$ is the statement that NED is conserved
through the layers.

\section{Two algebraic lemmas}\label{sec:lemmas}

Both lemmas below are purely algebraic. They use no property of the travel times,
and remain valid verbatim when $\alpha_k$ is replaced by an independent torus
variable in Section~\ref{sec:torus}.

\begin{lemma}[{\cite[Appendix B]{GSH2011}}]\label{lem:conj}
$D_k=C_k^{*}$ for all $k\ge0$.
\end{lemma}

\begin{proof}
For $k=0$, $C_0=D_0=1$. Assume $D_{k-1}=C_{k-1}^{*}$ and set
$\zeta_k:=\alpha_kC_{k-1}$. Since $\alpha_k^{*}C_{k-1}^{*}=\zeta_k^{*}$,
\eqref{eq:recursion} becomes
\begin{equation}
C_k=p_k\zeta_k+q_k\zeta_k^{*},\qquad D_k=q_k\zeta_k+p_k\zeta_k^{*} .
\label{eq:CD}
\end{equation}
As $p_k,q_k$ are real, $C_k^{*}=p_k\zeta_k^{*}+q_k\zeta_k=D_k$.
\end{proof}

\begin{lemma}[quadratic-form representation]\label{lem:quad}
With $\zeta_k=\alpha_kC_{k-1}$, for every $k\ge1$,
\begin{equation}
C_k=\Rp\zeta_k+iR_k\,\Ip\zeta_k,
\qquad
|C_k|^2=\bigl(\Rp\zeta_k\bigr)^2+R_k^2\bigl(\Ip\zeta_k\bigr)^2 .
\label{eq:quad}
\end{equation}
\end{lemma}

\begin{proof}
By \eqref{eq:CD} and \eqref{eq:pq},
$C_k=p_k\zeta_k+q_k\zeta_k^{*}
=(p_k+q_k)\Rp\zeta_k+i(p_k-q_k)\Ip\zeta_k
=\Rp\zeta_k+iR_k\Ip\zeta_k$.
\end{proof}

Lemma~\ref{lem:quad} generalizes \cite[Eq.~(30)]{GSH2011}, which was needed there
only at the lowermost interface, to every interface. All interfaces are required
for the peeling argument of Section~\ref{sec:main}.

\begin{corollary}[non-vanishing]\label{cor:pos}
Let $m_k=\min(1,R_k^2)$ and $M_k=\max(1,R_k^2)$. Then
$m_k|C_{k-1}|^2\le|C_k|^2\le M_k|C_{k-1}|^2$, and consequently
\begin{equation}
0<\prod_{j=1}^{n}m_j\;\le\;|C_n|^2\;\le\;\prod_{j=1}^{n}M_j<\infty
\label{eq:bounds}
\end{equation}
uniformly in $\omega$ (or in the torus variables).
\end{corollary}

\begin{proof}
Immediate from \eqref{eq:quad} and $|\zeta_k|=|C_{k-1}|$, using $C_0=1$.
\end{proof}

Corollary~\ref{cor:pos} states that no resonance drives the denominator to zero.
It is what makes $P_k$ a bounded continuous function on the torus, and hence what
licenses both Fubini's theorem and the use of unique ergodicity below. In
\cite{GSH2011} the corresponding fact was verified case by case; here it follows
at once, with a uniform bound.

\section{Lifting to the torus}\label{sec:torus}

Let $\T=\Real/2\pi\Zed$. For $\theta=(\theta_1,\dots,\theta_n)\in\T^n$ define
$C_k(\theta)$ and $D_k(\theta)$ by the recursion \eqref{eq:recursion} with
$\alpha_k$ replaced by $e^{i\theta_k}$. By construction,
\begin{equation}
C_k \text{ depends only on } \theta_1,\dots,\theta_k,
\label{eq:dependence}
\end{equation}
and not on $\theta_{k+1},\dots,\theta_n$. This triangular dependence is the
structural fact that makes the layer-peeling argument work; it is invisible as
long as everything is written as a function of the single variable $\omega$.

\begin{lemma}\label{lem:cont}
For $1\le k\le n$ the function
$F_k(\theta):=|C_{k-1}(\theta)|^2/|C_n(\theta)|^2$
is real, continuous and uniformly bounded on $\T^n$.
\end{lemma}

\begin{proof}
$C_j$ is a continuous function of $e^{i\theta_1},\dots,e^{i\theta_j}$, and by
Corollary~\ref{cor:pos} the denominator has a positive lower bound.
\end{proof}

The physical quantity is the restriction of $F_k$ to the line
$\ell_\tau=\{(\omega\tau_1,\dots,\omega\tau_n)\bmod 2\pi:\omega\in\Real\}$,
namely $P_k(\omega)=F_k(\omega\tau_1,\dots,\omega\tau_n)$.

\section{Main theorem: $\Rat$-independent travel times}\label{sec:main}

\subsection{The elementary integral}

\begin{lemma}\label{lem:elem}
For $R>0$,
$\displaystyle\frac{1}{2\pi}\int_0^{2\pi}
\frac{d\varphi}{\cos^2\varphi+R^2\sin^2\varphi}=\frac1R$.
\end{lemma}

\begin{proof}
The integrand has period $\pi$ and is symmetric about $\varphi=\pi/2$, so the
value is four times $\int_0^{\pi/2}$ divided by $2\pi$; and
$\int_0^{\pi/2}d\varphi/(\cos^2\varphi+R^2\sin^2\varphi)
=[R^{-1}\arctan(R\tan\varphi)]_0^{\pi/2}=\pi/(2R)$,
as in \cite[Eq.~(14)]{GSH2011}.
\end{proof}

This is exactly the integral used for two layers in \cite{GSH2011}. The whole
content of the general case is that it may be applied $n$ times in succession.

\subsection{Layer peeling}

\begin{lemma}[layer peeling]\label{lem:peel}
Let $1\le j\le n$ and let $G$ be any function not depending on $\theta_j$. Then
\begin{equation}
\frac{1}{2\pi}\int_0^{2\pi}\frac{G}{|C_j(\theta)|^2}\,d\theta_j
=\frac{1}{R_j}\cdot\frac{G}{|C_{j-1}(\theta)|^2}.
\end{equation}
\end{lemma}

\begin{proof}
By \eqref{eq:dependence}, $C_{j-1}$ does not depend on $\theta_j$. Write
$C_{j-1}=re^{i\psi}$ with $r=|C_{j-1}|>0$; then
$\zeta_j=e^{i\theta_j}C_{j-1}=re^{i\varphi}$ with $\varphi=\theta_j+\psi$, and
Lemma~\ref{lem:quad} gives
$|C_j|^2=r^2(\cos^2\varphi+R_j^2\sin^2\varphi)$. The Haar measure
$d\theta_j/2\pi$ is translation invariant, so the integral may be taken in
$\varphi$, and Lemma~\ref{lem:elem} applies.
\end{proof}

\subsection{The torus identity and the main theorem}

\begin{proposition}[torus identity]\label{prop:torus}
Let $1\le k\le n$. For any fixed $\theta_1,\dots,\theta_{k-1}$,
\begin{equation}
\frac{1}{(2\pi)^{\,n-k+1}}\int_{\T^{\,n-k+1}}
F_k(\theta)\,d\theta_k\cdots d\theta_n
=\prod_{j=k}^{n}\frac{1}{R_j}=\frac{1}{R_{k,0}} .
\label{eq:torusid}
\end{equation}
In particular the left-hand side does not depend on
$\theta_1,\dots,\theta_{k-1}$.
\end{proposition}

\begin{proof}
By Lemma~\ref{lem:cont} the integrand is continuous and bounded, so by Fubini's
theorem, I may integrate iteratively in the order
$\theta_n,\theta_{n-1},\dots,\theta_k$.

Start with $j=n$. The numerator $|C_{k-1}|^2$ does not depend on $\theta_n$
because $k-1<n$, so Lemma~\ref{lem:peel} with $G=|C_{k-1}|^2$ gives
\begin{equation}
\frac{1}{2\pi}\int_0^{2\pi}\frac{|C_{k-1}|^2}{|C_n|^2}\,d\theta_n
=\frac{1}{R_n}\cdot\frac{|C_{k-1}|^2}{|C_{n-1}|^2}.
\end{equation}
The result has the same form with $n$ replaced by $n-1$, and its numerator again
does not depend on $\theta_{n-1}$ when $k-1<n-1$. Iterating down to $j=k$,
\begin{equation}
\frac{1}{(2\pi)^{\,n-k+1}}\int_{\T^{\,n-k+1}}\frac{|C_{k-1}|^2}{|C_n|^2}
=\left(\prod_{j=k}^{n}\frac{1}{R_j}\right)\frac{|C_{k-1}|^2}{|C_{k-1}|^2}
=\prod_{j=k}^{n}\frac{1}{R_j},
\end{equation}
the numerator cancelling exactly at the last step. Now use
\eqref{eq:telescope}.
\end{proof}

\begin{theorem}[conservation of NED]\label{thm:main}
Assume $\tau_1,\dots,\tau_n$ are linearly independent over $\Rat$. Then for every
$k=1,\dots,n$,
\begin{equation}
\langle P_k\rangle=\frac{1}{R_{k,0}}=\frac{\rho_0\beta_0}{\rho_k\beta_k},
\qquad\text{so that}\qquad
\rho_k\beta_k\langle P_k\rangle=\rho_0\beta_0 .
\end{equation}
\end{theorem}

\begin{proof}
Linear independence of $\tau_1,\dots,\tau_n$ over $\Rat$ is equivalent, by Weyl's
equidistribution theorem \cite{Weyl1916}, to unique ergodicity of the linear flow
$\omega\mapsto(\omega\tau_1,\dots,\omega\tau_n)$ on $\T^n$: indeed, testing with a
character $e^{i\langle k,\theta\rangle}$, $k\in\Zed^n\setminus\{0\}$, gives
\begin{equation}
\frac1\Omega\int_{-\Omega/2}^{\Omega/2}e^{i\omega\langle k,\tau\rangle}\,d\omega
=\frac{2\sin\!\big(\Omega\langle k,\tau\rangle/2\big)}{\Omega\langle k,\tau\rangle}
=O\!\left(\frac{1}{\Omega\,|\langle k,\tau\rangle|}\right)\longrightarrow0,
\end{equation}
which is valid precisely when $\langle k,\tau\rangle\ne0$ for all
$k\ne0$. Unique ergodicity means that for continuous functions the time average
converges to the Haar average for \emph{every} initial point, uniformly. Since
$F_k$ is continuous by Lemma~\ref{lem:cont},
\begin{equation}
\langle P_k\rangle
=\lim_{\Omega\to\infty}\frac1\Omega\int_{-\Omega/2}^{\Omega/2}F_k(\omega\tau)\,d\omega
=\frac{1}{(2\pi)^n}\int_{\T^n}F_k\,d\theta .
\end{equation}
Applying Proposition~\ref{prop:torus} and then averaging the resulting constant
over $\theta_1,\dots,\theta_{k-1}$ gives $1/R_{k,0}$.
\end{proof}

The hypothesis holds for a set of full measure in $\Real^n$. In particular it is
satisfied with probability one by the Monte Carlo models of \cite[\S4]{GSH2011},
so Theorem~\ref{thm:main} rigorously justifies what those simulations indicated.

\begin{remark}[a stronger form]\label{rem:stronger}
Proposition~\ref{prop:torus} integrates only over $\theta_k,\dots,\theta_n$,
leaving the shallower phases fixed. Thus the NED of layer $\#k$ is obtained by
averaging over the phases of the layers \emph{below} it only; no assumption on
$\tau_1,\dots,\tau_{k-1}$ is needed. Precisely, if the orbit closure
$H\subset\T^n$ of the flow contains the subtorus
$\{0\}^{k-1}\times\T^{\,n-k+1}$, then $H$ is a union of cosets of that subtorus,
the Haar measure of $H$ disintegrates over them, and
Proposition~\ref{prop:torus} gives the same constant $1/R_{k,0}$ on each coset;
hence $\int_H F_k\,d\mathrm{Haar}_H=1/R_{k,0}$. This is a strictly stronger
statement than \eqref{eq:goal}.
\end{remark}

\section{Equal travel times}\label{sec:goupillaud}

The case $\tau_1=\cdots=\tau_n=\tau$ violates the hypothesis of
Theorem~\ref{thm:main} as strongly as possible: the orbit closure degenerates to
the diagonal circle in $\T^n$. It is also the standard discretization of an
arbitrary velocity profile (a Goupillaud medium \cite{Goupillaud}), so it deserves
a separate treatment. It turns out to follow from a classical fact about
orthogonal polynomials on the unit circle (OPUC).

\subsection{Polynomial structure}

Set $z=e^{i\omega\tau}$ and $u=z^2$. For a polynomial $\Psi$ of degree $m$ write
$\Psi^{R}(u):=u^{m}\bigl(\Psi(1/u^{*})\bigr)^{*}$ for its reversed polynomial.
As everywhere in this paper, $*$ denotes complex conjugation; note that in the
literature on orthogonal polynomials the superscript $*$ is often used for the
reversed polynomial itself, which I write $\Psi^{R}$ throughout.

\begin{lemma}\label{lem:poly}
$\Psi_k(u):=z^kC_k$ is a polynomial of degree $k$ in $u$ with
$|\Psi_k|=|C_k|$ on $|u|=1$, and
\begin{equation}
\Psi_k=p_k\,u\,\Psi_{k-1}+q_k\,\Psi_{k-1}^{R},\qquad \Psi_0=1 .
\label{eq:psirec}
\end{equation}
\end{lemma}

\begin{proof}
By \eqref{eq:CD} and Lemma~\ref{lem:conj},
$C_k=p_kzC_{k-1}+q_kz^{-1}C_{k-1}^{*}$. Multiplying by $z^k$ and using
$z^{k-1}C_{k-1}^{*}=\bigl(z^{-(k-1)}C_{k-1}\bigr)^{*}
=\bigl(u^{-(k-1)}\Psi_{k-1}\bigr)^{*}=u^{k-1}\Psi_{k-1}^{*}=\Psi_{k-1}^{R}$
on $|u|=1$ gives \eqref{eq:psirec}. The leading coefficient is multiplied by
$p_k>0$ at each step, so $\deg\Psi_k=k$.
\end{proof}

Normalizing, $\Phi_k:=\Psi_k/\prod_{j\le k}p_j$ is monic and satisfies the
Szeg\H{o} recursion
\begin{equation}
\Phi_k=u\,\Phi_{k-1}-a_k\,\Phi_{k-1}^{R},
\qquad a_k:=-\frac{q_k}{p_k}\in(-1,1),
\label{eq:szego}
\end{equation}
so that $a_k$ is the $k$th Verblunsky coefficient. By \eqref{eq:pq},
\begin{equation}
1-a_k^2=\frac{p_k^2-q_k^2}{p_k^2}=\frac{R_k}{p_k^2}.
\label{eq:verb}
\end{equation}

\subsection{Proof via the Bernstein--Szeg\H{o} normalization}

\begin{theorem}\label{thm:goupillaud}
If $\tau_1=\cdots=\tau_n=\tau$, then $\langle P_k\rangle=1/R_{k,0}$ for every
$k=1,\dots,n$.
\end{theorem}

\begin{proof}
As $\omega$ ranges over $\Real$, $\theta=\omega\tau$ is uniformly distributed on
the circle, and so is $u=e^{2i\theta}$ (it wraps twice). Hence
\begin{equation}
\langle P_k\rangle
=\frac{1}{2\pi}\int_0^{2\pi}\frac{|\Psi_{k-1}|^2}{|\Psi_n|^2}\,d\arg u
=\frac{\prod_{j\le k-1}p_j^2}{\prod_{j\le n}p_j^2}\cdot
\frac{1}{2\pi}\int_0^{2\pi}\frac{|\Phi_{k-1}|^2}{|\Phi_n|^2}\,d\arg u .
\label{eq:gp1}
\end{equation}
By Lemma~\ref{lem:appmain} of Appendix~\ref{app:opuc}, applied with $j=k-1$,
\begin{equation}
\frac{1}{2\pi}\int_0^{2\pi}\frac{|\Phi_{k-1}|^2}{|\Phi_n|^2}\,d\arg u
=\prod_{j=k}^{n}\frac{1}{1-a_j^2}.
\end{equation}
Substituting into \eqref{eq:gp1} and using \eqref{eq:verb},
\begin{equation}
\langle P_k\rangle=\frac{1}{\prod_{j=k}^{n}p_j^2}
\prod_{j=k}^{n}\frac{p_j^2}{R_j}
=\prod_{j=k}^{n}\frac{1}{R_j}=\frac{1}{R_{k,0}} . \qedhere
\end{equation}
\end{proof}

\begin{remark}
For $k=1$, \eqref{eq:gp1} reduces to
$\int d\arg u/(2\pi|\Phi_n|^2)=\prod_j(1-a_j^2)^{-1}$, which is the classical
Bernstein--Szeg\H{o} normalization, a fact going back to Szeg\H{o}'s papers of
1920--21 \cite{Szego1975}; see also \cite{Simon2005}. Appendix~\ref{app:opuc}
gives a short self-contained proof in the present notation, so that no external
result is needed. The Szeg\H{o} product alone is not the
impedance ratio: only together with the prefactor does one get
$\prod_j p_j^2(1-a_j^2)=\prod_j R_j=R_{1,0}$. Thus the equal-travel-time case of
NED conservation is a restatement of a classical identity, though not through a
naive one-to-one correspondence.
\end{remark}

\section{Reduction by subdivision, and the remaining case}\label{sec:reduction}

Theorems~\ref{thm:main} and~\ref{thm:goupillaud} cover the two extremes:
$\Rat$-independent travel times, and all travel times equal. In this section, I
show that subdivision of layers by reflectionless interfaces extends
Theorem~\ref{thm:goupillaud} to \emph{all commensurable} travel times, and
reduces the general case to a canonical form. This section replaces a conjecture
stated in an earlier draft of this note, which the reduction below renders
unnecessary in the commensurable case.

\subsection{Reflectionless subdivision}

\begin{lemma}[transparency]\label{lem:transparent}
Suppose layer $\#j$ and layer $\#(j+1)$ have identical material properties, so
that $R_j=1$. Then $p_j=1$, $q_j=0$, $C_j=e^{i\theta_j}C_{j-1}$ and in particular
$|C_j|=|C_{j-1}|$. Consequently, splitting a layer of travel time
$\tau=\tau'+\tau''$ into two sublayers of the same material with travel times
$\tau'$ and $\tau''$ changes neither $|C_j|$ at any original interface nor any
$\langle P_k\rangle$.
\end{lemma}

\begin{proof}
Immediate from \eqref{eq:CD} with $p_j=1$, $q_j=0$. The composite propagator
across the two sublayers is $e^{i\omega(\tau'+\tau'')}$ times the identity on the
relevant factor, which is the propagator of the undivided layer.
\end{proof}

Note also that a fictitious interface is harmless in the OPUC picture: $R_j=1$
corresponds to the Verblunsky coefficient $a_j=0$, which is admissible
($|a_j|<1$), and \eqref{eq:verb} reads $1-0=1=R_j/p_j^2$.

\begin{theorem}[commensurable travel times]\label{thm:commensurable}
Suppose there are $\tau_0>0$ and positive integers $m_1,\dots,m_n$ with
$\tau_k=m_k\tau_0$ for all $k$. Then $\langle P_k\rangle=1/R_{k,0}$ for every
$k=1,\dots,n$.
\end{theorem}

\begin{proof}
Subdivide layer $\#k$ into $m_k$ sublayers of the same material, each of travel
time $\tau_0$. By Lemma~\ref{lem:transparent} all $\langle P_k\rangle$ are
unchanged, and the uppermost sublayer of the original layer $\#k$ carries the
amplitude $A_k$ and the impedance $Z_k$. The subdivided structure has
$n'=\sum_k m_k$ layers, all with travel time $\tau_0$, so
Theorem~\ref{thm:goupillaud} applies to it and gives
$Z'_j\langle P'_j\rangle=Z_0$ for every sublayer $\#j$. Reading this at the
uppermost sublayer of each original layer yields the claim.
\end{proof}

\subsection{Canonical form for the general case}

Commensurability is strictly stronger than $\Rat$-linear dependence: for
$n=3$ and $\tau=(1,\sqrt2,1+\sqrt2)$ the travel times satisfy the relation
$\tau_3=\tau_1+\tau_2$, yet no common $\tau_0$ exists. Nevertheless a subdivision
is still available, at the price of using several distinct sublayer travel times.

\begin{lemma}[positive integral basis]\label{lem:cone}
Let $\tau_1,\dots,\tau_n>0$ span a $\Rat$-vector space $V\subset\Real$ of
dimension $d$. Then there exist $\nu_1,\dots,\nu_d>0$, linearly independent over
$\Rat$, and non-negative integers $m_{ki}$ such that
\begin{equation}
\tau_k=\sum_{i=1}^{d}m_{ki}\,\nu_i\qquad(k=1,\dots,n).
\end{equation}
\end{lemma}

\begin{proof}
For $d=1$ this is commensurability, so assume $d\ge2$. Fix a $\Rat$-basis of $V$
and let $\iota:\Rat^d\to\Real$, $\iota(x)=\langle x,\lambda\rangle$, be the
induced injection, where $\lambda\in\Real^d$ collects the basis values. Let
$v_k\in\Rat^d$ be the coordinate vector of $\tau_k$, so
$\langle v_k,\lambda\rangle=\tau_k>0$. Thus all $v_k$ lie in the open half-space
$H=\{x\in\Real^d:\langle x,\lambda\rangle>0\}$, and
$C:=\cone(v_1,\dots,v_n)$ satisfies $C\setminus\{0\}\subset H$; in particular $C$
is pointed.

Choose $f\in\Rat^d$ close to $\lambda$, so that $\langle v_k,f\rangle>0$ for all
$k$, and put $P:=\mathrm{conv}\{v_k/\langle v_k,f\rangle\}$, a compact set in the
rational hyperplane $\{\langle\cdot,f\rangle=1\}$. Within that hyperplane,
\begin{equation}
U:=H\cap\{\langle\cdot,f\rangle=1\}
=\{x:\langle x,f\rangle=1,\ \langle x,\lambda-f\rangle>-1\}
\end{equation}
is an open half-space whose boundary lies at distance
$\asymp\|\lambda-f\|^{-1}$ from $P$. Taking $f$ close enough to $\lambda$, $U$
contains a Euclidean ball around $P$ of radius exceeding $d\cdot\mathrm{diam}(P)$,
hence contains a $(d-1)$-simplex $T\supseteq P$; perturbing the vertices of $T$
slightly, they may be taken rational, with $T\supseteq P$ and $T\subset U$ still.

Let $w_1,\dots,w_d\in\Rat^d$ be the vertices of $T$. They form a basis of
$\Rat^d$, and $\cone(w_1,\dots,w_d)\supseteq C$, so each $v_k=\sum_i c_{ki}w_i$
with $c_{ki}\in\Rat_{\ge0}$. Let $N$ be a common denominator of the $c_{ki}$ and
set $\nu_i:=\langle w_i,\lambda\rangle/N$ and $m_{ki}:=Nc_{ki}\in\Zed_{\ge0}$.
Then $\nu_i>0$ since $w_i\in U\subset H$, the $\nu_i$ are $\Rat$-independent since
$\iota$ is injective and the $w_i$ form a basis, and
$\tau_k=\langle v_k,\lambda\rangle=\sum_i m_{ki}\nu_i$.
\end{proof}

\begin{example}
For $\tau=(1,\sqrt2,1+\sqrt2)$ take $\nu=(1,\sqrt2)$, with coefficient vectors
$(1,0),(0,1),(1,1)$. For $\tau=(\sqrt2,\sqrt3,\sqrt3-\sqrt2)$ the naive basis
$(\sqrt2,\sqrt3)$ gives the inadmissible coefficients $(-1,1)$ for the third
entry, but $\nu=(\sqrt3-\sqrt2,\sqrt2)$ gives $(0,1),(1,1),(1,0)$. For
$\tau=(1,\sqrt2,3-\sqrt2)$ one may take
$\nu=\bigl(\tfrac{\sqrt2}{3},\tfrac{3-\sqrt2}{3}\bigr)$ with coefficients
$(1,1),(3,0),(0,3)$.
\end{example}

Combining Lemmas~\ref{lem:transparent} and~\ref{lem:cone}:

\begin{proposition}[canonical form]\label{prop:canonical}
Every layered structure is equivalent, after subdivision by reflectionless
interfaces, to one whose travel times take only $d$ distinct values
$\nu_1,\dots,\nu_d$, linearly independent over $\Rat$, where $d=\dim_\Rat\langle
\tau_1,\dots,\tau_n\rangle$. The values $\langle P_k\rangle$ are unchanged.
\end{proposition}

In this canonical form, the phase of every sublayer is one of $d$ independent
variables $\varphi_i=\omega\nu_i$, and the flow
$\omega\mapsto(\varphi_1,\dots,\varphi_d)$ is equidistributed on $\T^d$ by Weyl's
theorem. The two theorems proved above are the endpoints:
\begin{itemize}
\item $d=1$: all sublayers share a single phase. Theorem~\ref{thm:commensurable}.
\item $d=n'$: no phase is repeated. Theorem~\ref{thm:main}.
\end{itemize}

\subsection{The open case}

\begin{problem}\label{prob:open}
Prove \eqref{eq:goal} for a structure in canonical form with $1<d<n'$, that is,
when the $d$ independent phase variables are shared by more than $d$ sublayers.
\end{problem}

The single obstruction is identified precisely. The layer-peeling
Lemma~\ref{lem:peel} relies on \eqref{eq:dependence}: $C_{j-1}$ must not depend on
the variable being integrated. When two sublayers share a phase variable this
fails, and the iteration of Proposition~\ref{prop:torus} breaks down after the
first repetition. Equivalently, expanding as in Lemma~\ref{lem:quad} with
$u_j:=\zeta_j/\zeta_j^{*}$ and using $p_j>|q_j|$, the Poisson kernel expansion
\begin{equation}
\frac{1}{|C_j|^2}
=\frac{1}{R_j}\cdot\frac{1}{|C_{j-1}|^2}
\sum_{m\in\Zed}\left(-\frac{q_j}{p_j}\right)^{|m|}u_j^{\,m}
\label{eq:poisson}
\end{equation}
has $m=0$ term equal to the peeling result, and one must show that the
contributions of the terms $m\ne0$ average to zero. This is automatic when the
variable is not repeated, and follows indirectly from
Theorem~\ref{thm:goupillaud} when $d=1$.

A $d$-variable analogue of the Bernstein--Szeg\H{o} theory---orthogonal
polynomials on the distinguished boundary of the polydisk---appears to be the
natural setting; the polydisk framework developed by Gibson \cite{Gibson2013} for
layered scattering may be relevant, although his results address a different
quantity (see Section~\ref{sec:priorwork}). I note that the set of travel-time
vectors with $1<d<n$ has Lebesgue measure zero, and that any structure can be
approximated to arbitrary accuracy by a commensurable one, so
Problem~\ref{prob:open} does not obstruct applications. Numerically the identity
holds in the intermediate case as well (Section~\ref{sec:numerics}).

\section{Oblique incidence}\label{sec:oblique}

\cite[Appendix C]{GSH2011} shows that if the incidence angle $\delta$ and the
velocities satisfy $\beta_0>\beta_k\sin\delta$ in every layer, then oblique
incidence reduces to the vertical problem upon replacing
\begin{equation}
\beta_k\ \longrightarrow\ c_k=\beta_k\Bigl(1-\tfrac{\beta_k^2}{\beta_0^2}
\sin^2\delta\Bigr)^{-1/2},
\qquad
Z_k\ \longrightarrow\ \tilde Z_k=\frac{\rho_k\beta_k^2}{c_k}.
\end{equation}
Under that condition $c_k$ is real and $\tilde R_k=\tilde Z_k/\tilde Z_{k+1}$ is a
positive real, so \eqref{eq:recursion} retains exactly the form used here. Every
lemma, proposition and theorem above therefore applies verbatim with
$R_k\to\tilde R_k$ and $\tau_k\to h_k/c_k$, yielding
$\tilde Z_k\langle P_k\rangle=\tilde Z_0$. The independence hypotheses are to be
imposed on $h_k/c_k$.

\section{Attenuation}\label{sec:damping}

Everything above assumes elasticity. With attenuation, $R_k$ becomes complex,
Lemma~\ref{lem:quad} no longer represents $|C_k|^2$ as a real quadratic form in
$(\Rp\zeta_k,\Ip\zeta_k)$, and layer peeling fails. This is not an artefact of the
method: NED is known to decrease with damping, and that decrease is precisely
what makes it useful for the direct estimation of near-surface damping
\cite{Goto2013}. The present results fix the elastic baseline against which such
estimates are made.

\section{Numerical verification}\label{sec:numerics}

Each step was checked numerically. Following \cite[\S4]{GSH2011}, S-wave
velocities were drawn uniformly from $[10,700]\,\mathrm{m/s}$ and densities from
$[1000,2000]\,\mathrm{kg/m^3}$ for every layer and for the basement. For each of
20 samples per setting, I evaluated
$\max_k\bigl|Z_k\langle P_k\rangle/Z_0-1\bigr|$.

\begin{center}
\begin{tabular}{lccc}
\toprule
travel times & $n=2$ & $n=3$ & $n=5$\\
\midrule
uniform random ($\Rat$-independent) & $3.8\times10^{-4}$ & $9.6\times10^{-4}$ & $6.3\times10^{-3}$\\
equal, $\tau_k\equiv\tau$ & $2.9\times10^{-15}$ & $1.7\times10^{-15}$ & $4.0\times10^{-15}$\\
commensurable, integer ratios & $1.6\times10^{-15}$ & $1.3\times10^{-15}$ & $2.9\times10^{-15}$\\
\bottomrule
\end{tabular}
\end{center}

For the periodic settings an equispaced quadrature over one period is exact up to
rounding, and the residuals are at the level of double precision. For
$\Rat$-independent travel times the frequency integral must be truncated, and the
residual is dominated by that truncation. For $n=10$ the required resolution grows
rapidly; convergence was verified instead, e.g.\ for one commensurable ten-layer
sample the residual fell from $2.7\times10^{-2}$ to $2.0\times10^{-5}$ when the
number of quadrature points per period was increased from $2\times10^6$ to
$8\times10^6$, confirming that the discrepancy is a quadrature artefact.

The intermediate case of Problem~\ref{prob:open} was tested with $n=3$ and
$\tau=(1,\sqrt2,1+\sqrt2)$ and $\tau=(\sqrt2,\sqrt3,\sqrt3-\sqrt2)$ (both $d=2$);
residuals were $2.7\times10^{-6}$ and $8.4\times10^{-6}$ respectively, comparable
to the $\Rat$-independent control computed with the same truncation.

Finally, the stronger statement of Remark~\ref{rem:stronger} was verified for
$n=6$, $k=4$: fixing $\theta_1,\theta_2,\theta_3$ at three different values and
averaging over $\theta_4,\theta_5,\theta_6$ only reproduced $1/R_{4,0}$ in each
case. The average over the remaining $3$-torus was computed with a
tensor-product trapezoidal rule, which converges spectrally for this smooth
periodic integrand; with $128$ nodes per axis the three values agreed with
$1/R_{4,0}$ to within $10^{-8}$, independently of the fixed shallow phases.

\section{Relation to previous work}\label{sec:priorwork}

\subsection*{Equal travel times}
As shown in Section~\ref{sec:goupillaud}, NED conservation for equal travel times
is a consequence of the Bernstein--Szeg\H{o} normalization, which dates from
Szeg\H{o}'s work of 1920--21 \cite{Szego1975}. The dictionary between
equal-travel-time layered media and the Szeg\H{o} recursion---reflection
coefficients corresponding to Verblunsky coefficients---was established in the
1970s and 1980s \cite{BrucksteinKailath,UrsinBerteussen}. In this case, therefore,
all the mathematics required was available long before \cite{GSH2011}. I have not
found the identity stated anywhere as a conservation law through the layers.

\subsection*{General travel times}
The lifting of Section~\ref{sec:torus} coincides with the framework used by
Gibson \cite{Gibson2013} for scattering in layered media, who shows that the
restriction of a Hardy function on the polydisk to a generic line on the torus is
Besicovitch almost periodic with Besicovitch norm equal to the $H^2$ norm. That
statement plays the role that Weyl's theorem plays here.

The objects and the assertions differ, however. Gibson studies the reflection
Green's function of a stack sandwiched between two half-spaces; the boundary
condition is not a free surface. In his parametrization the present configuration
is the boundary point $R=(\dots,1)$, where $|\widehat G|\equiv1$ holds exactly, so
his main theorem---a probabilistic estimate showing the power spectrum to be
approximately constant for many layers---is vacuous there. Moreover
$|A_k/A_0|^2$ concerns waves circulating \emph{inside} the medium: because they
are totally reflected at the free surface they never leave, so the quantity is not
determined by the scattering balance $\sum a_j^2+\sum b_j^2=1$ of
\cite[Eq.~(1.10)]{Gibson2013}. Finally, no analogue of the layer peeling of
Lemma~\ref{lem:peel} occurs there.

\subsection*{Summary}
I have found no prior statement or proof of NED conservation. On the other hand
the tools used here were all available when \cite{GSH2011} was written, and the
equal-travel-time case is a direct consequence of classical results. The
contribution of \cite{GSH2011} is therefore best regarded as the identification of
the quantity and its formulation as a conservation law with seismological
applications, rather than as a technically difficult theorem; the present note
supplies the proof that was missing.

\appendix

\section{A self-contained proof of the identity used in Section~\ref{sec:goupillaud}}
\label{app:opuc}

The identity needed in the proof of Theorem~\ref{thm:goupillaud} is classical: it is
the Bernstein--Szeg\H{o} normalization together with the orthogonality of the
$\Phi_j$ with respect to the associated measure, in the theory of orthogonal
polynomials on the unit circle \cite{Szego1975,Simon2005}. Since only one special
case is used here, and since it costs half a page, I give a direct proof in the
notation of Section~\ref{sec:goupillaud}, so that the paper is self-contained.

Throughout, $\Disk=\{u:|u|<1\}$, and for a polynomial $\Phi$ of degree $m$ we write
$\Phi^{R}(u)=u^{m}\bigl(\Phi(1/u^{*})\bigr)^{*}$ for its reversal, so that
\begin{equation}
\Phi^{R}(u)=u^{m}\,\Phi(u)^{*},\qquad\text{hence}\qquad
|\Phi^{R}(u)|=|\Phi(u)|\qquad (|u|=1).
\label{eq:revcircle}
\end{equation}
The polynomials $\Phi_k$ are those of \eqref{eq:szego}: $\Phi_0=1$ and
\begin{equation}
\Phi_k=u\,\Phi_{k-1}-a_k\,\Phi_{k-1}^{R},
\qquad a_k=-\frac{q_k}{p_k}\in(-1,1),
\label{eq:appszego}
\end{equation}
each $\Phi_k$ monic of degree $k$. The coefficients $a_k$ are real, which
simplifies the formulas slightly.

\begin{lemma}[reversed recursion]\label{lem:apprev}
For $k\ge1$,
\begin{equation}
\Phi_k^{R}=\Phi_{k-1}^{R}-a_k\,u\,\Phi_{k-1},
\label{eq:apprev}
\end{equation}
and $\Phi_k^{R}(0)=1$ for all $k\ge0$.
\end{lemma}

\begin{proof}
Apply $\Phi\mapsto u^{k}\bigl(\Phi(1/u^{*})\bigr)^{*}$ to \eqref{eq:appszego}. Since
$a_k$ is real,
\begin{equation}
u^{k}\bigl(\Phi_k(1/u^{*})\bigr)^{*}
=u^{k-1}\bigl(\Phi_{k-1}(1/u^{*})\bigr)^{*}
-a_k\,u^{k}\bigl(\Phi_{k-1}^{R}(1/u^{*})\bigr)^{*} .
\end{equation}
The first term on the right is $\Phi_{k-1}^{R}$ by definition. For the second,
$\Phi_{k-1}^{R}(1/u^{*})=(1/u^{*})^{k-1}\Phi_{k-1}(u)^{*}$, so
$u^{k}\bigl(\Phi_{k-1}^{R}(1/u^{*})\bigr)^{*}=u\,\Phi_{k-1}$, giving
\eqref{eq:apprev}. Finally $\Phi_k^{R}(0)$ is the complex conjugate of the
leading coefficient of $\Phi_k$, which is $1$ because $\Phi_k$ is monic.
\end{proof}

\begin{lemma}[the reversed polynomials are zero-free]\label{lem:appzero}
For every $k\ge0$, $\Phi_k^{R}(u)\ne0$ for all $u\in\overline{\Disk}$.
\end{lemma}

\begin{proof}
Induction on $k$; the case $k=0$ is trivial since $\Phi_0^{R}=1$. Suppose
$\Phi_{k-1}^{R}$ has no zero in $\overline{\Disk}$. Then
\begin{equation}
f(u):=\frac{u\,\Phi_{k-1}(u)}{\Phi_{k-1}^{R}(u)}
\end{equation}
is analytic on $\overline{\Disk}$, and $|f|=1$ on $|u|=1$ by
\eqref{eq:revcircle}. By the maximum principle $|f|\le1$ on $\overline{\Disk}$.
Lemma~\ref{lem:apprev} gives $\Phi_k^{R}=\Phi_{k-1}^{R}\,(1-a_kf)$, and
$|a_kf|\le|a_k|<1$, so the second factor does not vanish; neither does the first,
by the inductive hypothesis.
\end{proof}

\begin{lemma}\label{lem:appmain}
For $0\le j\le n$,
\begin{equation}
\frac{1}{2\pi}\int_0^{2\pi}\frac{|\Phi_j(u)|^2}{|\Phi_n(u)|^2}\,d\arg u
=\prod_{i=j+1}^{n}\frac{1}{1-a_i^{2}},
\label{eq:appmain}
\end{equation}
with the convention that an empty product equals $1$.
\end{lemma}

\begin{proof}
Fix $j$ and induct downwards on $n\ge j$. For $n=j$ both sides equal $1$. Let
$n>j$ and assume \eqref{eq:appmain} holds with $n-1$ in place of $n$.

On $|u|=1$ put $s:=u\Phi_{n-1}/\Phi_{n-1}^{R}$, so $|s|=1$ by
\eqref{eq:revcircle}, and \eqref{eq:appszego} reads
$\Phi_n=\Phi_{n-1}^{R}\,(s-a_n)$. Hence
\begin{equation}
|\Phi_n|^2=|\Phi_{n-1}|^2\,|s-a_n|^2,
\qquad\text{so}\qquad
\frac{|\Phi_j|^2}{|\Phi_n|^2}
=\frac{|\Phi_j|^2}{|\Phi_{n-1}|^2}\cdot\frac{1}{|s-a_n|^2}.
\end{equation}
Since $|a_n|<1$ and $|s|=1$, the Poisson kernel expansion---the same one used in
\eqref{eq:poisson}---gives
\begin{equation}
\frac{1}{|s-a_n|^2}=\frac{1}{1-a_n^{2}}\sum_{m\in\Zed}a_n^{|m|}\,s^{\,m}.
\label{eq:apppoisson}
\end{equation}
The term $m=0$ contributes $\frac{1}{1-a_n^{2}}$ times the integral appearing in
the inductive hypothesis, so it suffices to show that
\begin{equation}
\frac{1}{2\pi}\int_0^{2\pi}\frac{|\Phi_j|^2}{|\Phi_{n-1}|^2}\,s^{\,m}\,d\arg u=0
\qquad(m\ne0).
\label{eq:appcross}
\end{equation}

Let $m\ge1$. Using $\Phi_l^{*}=u^{-l}\Phi_l^{R}$ on $|u|=1$, which is
\eqref{eq:revcircle} again,
\begin{equation}
\frac{|\Phi_j|^2}{|\Phi_{n-1}|^2}\,s^{\,m}
=u^{\,n-1-j}\,\frac{\Phi_j\Phi_j^{R}}{\Phi_{n-1}\Phi_{n-1}^{R}}
\cdot\frac{u^{m}\Phi_{n-1}^{m}}{(\Phi_{n-1}^{R})^{m}}
=u^{\,n-1-j+m}\,
\frac{\Phi_j\,\Phi_j^{R}\,\Phi_{n-1}^{\,m-1}}{(\Phi_{n-1}^{R})^{\,m+1}}
=:G_m(u).
\end{equation}
The numerator of $G_m$ is a polynomial, and by Lemma~\ref{lem:appzero} the
denominator has no zero in $\overline{\Disk}$; hence $G_m$ is analytic on
$\overline{\Disk}$. Moreover $n-1-j+m\ge1$ because $j\le n-1$ and $m\ge1$, so
$G_m(0)=0$. The mean value property therefore gives
\begin{equation}
\frac{1}{2\pi}\int_0^{2\pi}G_m(u)\,d\arg u=G_m(0)=0 .
\end{equation}
For $m\le-1$, note that $|\Phi_j|^2/|\Phi_{n-1}|^2$ is real and
$(s^{\,m})^{*}=s^{-m}$, so the integral in \eqref{eq:appcross} is the complex
conjugate of the corresponding integral with $-m\ge1$, hence also zero. This
proves \eqref{eq:appcross} and completes the induction.
\end{proof}

Taking $j=0$ in \eqref{eq:appmain} shows that
$\|\Phi_n\|^{2}\,d\arg u/(2\pi|\Phi_n|^{2})$ with
$\|\Phi_n\|^{2}:=\prod_{i=1}^{n}(1-a_i^{2})$ is a probability measure, and the
general case expresses the orthogonality of the $\Phi_j$ with respect to it. The
proof above uses only the maximum principle, the mean value property and the
Poisson kernel.

\section*{Acknowledgements}

The Normalized Energy Density was introduced jointly with Sawada and
Hirai in \cite{GSH2011}, and this note builds directly on that work. 
The proofs presented here were developed in dialogue with Claude, an AI assistant
made by Anthropic. In particular, the lifting to the torus with Weyl's theorem
(Section~\ref{sec:torus}), the Bernstein--Szeg\H{o} route for the
equal-travel-time case (Section~\ref{sec:goupillaud}), and the numerical
verification of Section~\ref{sec:numerics} originated in that exchange; the
reduction by reflectionless subdivision (Section~\ref{sec:reduction}) originated
with the author. In accordance with arXiv policy, the tool is not listed as an
author. The author has verified all arguments and takes full responsibility for
the content, including any errors.



\begin{thebibliography}{99}

\bibitem{GSH2011}
H.~Goto, S.~Sawada and T.~Hirai,
\emph{Conserved quantity of elastic waves in multi-layered media: 2D SH case ---
Normalized Energy Density},
Wave Motion \textbf{48} (2011), 603--613.

\bibitem{Goto2013}
H.~Goto, Y.~Kawamura, S.~Sawada and T.~Akazawa,
\emph{Direct estimation of near-surface damping based on normalized energy
density},
Geophys.\ J.\ Int.\ \textbf{194} (2013), 488--498.

\bibitem{KokushoMotoyama}
T.~Kokusho and R.~Motoyama,
\emph{Energy dissipation in surface layer due to vertically propagating SH wave},
J.\ Geotech.\ Geoenviron.\ Eng.\ \textbf{128} (2002), 309--318.

\bibitem{Goupillaud}
P.~L. Goupillaud,
\emph{An approach to inverse filtering of near-surface layer effects from seismic
records},
Geophysics \textbf{26} (1961), 754--760.

\bibitem{Weyl1916}
H.~Weyl,
\emph{\"Uber die Gleichverteilung von Zahlen mod.\ Eins},
Math.\ Ann.\ \textbf{77} (1916), 313--352.

\bibitem{Szego1975}
G.~Szeg\H{o},
\emph{Orthogonal Polynomials},
4th ed., AMS Colloquium Publications XXIII, Providence, RI, 1975.

\bibitem{Simon2005}
B.~Simon,
\emph{Orthogonal Polynomials on the Unit Circle, Part 1: Classical Theory},
AMS Colloquium Publications 54, Providence, RI, 2005.

\bibitem{BrucksteinKailath}
A.~M. Bruckstein and T.~Kailath,
\emph{Inverse scattering for discrete transmission-line models},
SIAM Review \textbf{29} (1987), 359--389.

\bibitem{UrsinBerteussen}
B.~Ursin and K.-A. Berteussen,
\emph{Comparison of some inverse methods for wave propagation in layered media},
Proc.\ IEEE \textbf{74} (1986), 389--400.

\bibitem{Gibson2013}
P.~C. Gibson,
\emph{Hardy space on the polydisk and scattering in layered media},
preprint, arXiv:1306.2871 (2013).

\bibitem{FGPS}
J.-P. Fouque, J.~Garnier, G.~Papanicolaou and K.~S{\o}lna,
\emph{Wave Propagation and Time Reversal in Randomly Layered Media},
Springer, New York, 2007.

\bibitem{AkiRichards}
K.~Aki and P.~G. Richards,
\emph{Quantitative Seismology},
2nd ed., University Science Books, Sausalito, CA, 2002.

\end{thebibliography}
\end{document}